\documentclass[conference,letterpaper]{IEEEtran}

\usepackage{amsthm,amsfonts,braket,algorithm,graphicx,subfigure}
\usepackage{balance}

\usepackage[utf8]{inputenc} 
\usepackage[T1]{fontenc}
\usepackage{url}
\usepackage{ifthen}
\usepackage{cite}
\usepackage[cmex10]{amsmath} 

\theoremstyle{definition} 
\theoremstyle{definition} 
\newtheorem {theorem} {Theorem}

\newtheorem {lemma} {Lemma}

\newcommand{\kb}[1]{\mathbf{\left[#1\right]}}

\newcommand{\al}{\mathcal{A}}

\newcommand{\trd}[1]{\left|\left| #1 \right| \right|}

\newcommand{\samp}{\Phi}

\newcommand{\st}{\text{ } | \text{ }}

\newcommand{\Hmin}{H_\infty}
\newcommand{\Hextd}{\bar{H}}

\newcommand{\experiment}[1]{\leftarrow\textbf{\texttt{Exp}}\left(#1\right)}
\newcommand{\experimentname}{\textbf{\texttt{Exp}}}

\floatname{algorithm}{Protocol}

\begin{document}
\title{A New High-Dimensional Quantum Entropic Uncertainty Relation with Applications}

\author{%
  \IEEEauthorblockN{Walter O. Krawec}
  \IEEEauthorblockA{University of Connecticut\\
                    Department of Computer Science and Engineering\\
                    Storrs, CT, USA 06269\\
                    Email: walter.krawec@gmail.com}
}

\maketitle

\begin{abstract}
In this paper we derive a new quantum entropic uncertainty relation, bounding the conditional smooth quantum min entropy based on the result of a measurement using a two outcome POVM and the failure probability of a classical sampling strategy.  Our relation works for systems of arbitrary dimension.  We apply it to analyze a new source independent quantum random number generation protocol and show our relation provides optimistic results compared to prior work.
\end{abstract}
\emph{This is a (slightly) extended version of a paper to appear in IEEE ISIT 2020.}
\section{Introduction}
Quantum entropic uncertainty relations have numerous applications in quantum information, communication, and cryptography.  Informally, typical relations of this kind bound the amount of uncertainty in two different measurements performed on a quantum system.  This bound is typically a function of the overlap between the measurements performed.  Though there are many varieties \cite{MU-bound,smooth-uncertainty,ent2,ent3,ent1,ent4} (just to list a few - see \cite{survey,survey-2,survey-3} for a general survey).

Conditional quantum min entropy (which we define formally later but denote $\Hmin(A|E)$) is a very useful resource in quantum cryptography \cite{renner2008security} and so discovering new uncertainty bounds involving the min entropy of a system is important in various applications (though, outside of applications, such bounds are also interesting in and of themselves).  For instance, a useful quantum min entropy uncertainty relation was shown in \cite{smooth-uncertainty} and states that $\Hmin^\epsilon(Z|E) + H_{\max}^\epsilon(X|B) \ge \gamma$, where $\gamma$ is a function of the overlap of the two measurements (used to produce registers $Z$ and $X$ respectively) and $H_{\max}$ is the max entropy \cite{renner2004smooth,konig2009operational,renner2008security}.  Such a relation may be used, for instance, to bound an adversary's uncertainty on a quantum system given that the $X$ and $B$ registers are highly correlated.

In this work, we introduce a new quantum uncertainty relation, bounding the conditional quantum min entropy of a system based on the Hamming weight of a measurement outcome performed using a two-outcome POVM and the error probability of a \emph{classical} sampling technique.  Our relation applies to systems of arbitrary, but known and finite, dimension.  To our knowledge this form of uncertainty relation has not been discovered before.  To prove our relation, we utilize a quantum sampling framework introduced by Bouman and Fehr in \cite{sampling}.  This sampling framework was used in \cite{sampling} to prove the security of BB84.  Only recently, we discovered in \cite{krawec2019quantum} that it can be extended to more general areas of quantum information theory.  In particular we proved a quantum entropic uncertainty relation, however our previous relation from \cite{krawec2019quantum} was only applicable to qubits (dimension two systems) and did not involve the conditional min entropy.  As we consider conditional entropy here, our new bound is immediately applicable to quantum cryptographic applications.  We demonstrate this by considering and analyzing a new high-dimensional source independent quantum random number generator (QRNG).  Thanks to our new entropic uncertainty relation, and in particular it's need for only a two-outcome POVM in one of the measurements, our new QRNG does not require a full basis measurement in the test case making it potentially more practical (though, we stress, we are not interested in practical issues in this paper, only theoretical analyses).  We show that our new bound provides very optimistic random number generation rates when compared to other high dimensional QRNG's, even considering our protocol's simplicity in its quantum capabilities.

Our main result is described formally in Theorem \ref{thm:main}.  At a high level, our main result shows that for a given quantum state $\rho_{AE}$ (which is not necessarily i.i.d.), where the $A$ register acts on $n+m$ copies of a $d$-dimensional Hilbert space, if one were to measure part of the $A$ system using a particular two-outcome POVM, then, with high probability, one can bound the min entropy \emph{in the remaining unmeasured portion of the partially measurement state} should a measurement in a $d$ dimensional basis be performed on the remaining system.  This bound is a function of the observed outcome of the POVM measurement (in particular, the Hamming weight of this outcome) and also a function of the measurements performed.  This has interesting cryptographic applications as it allows one to argue about the entropy in partially measured states given a particular measurement outcome, with high probability.  Due to the two-outcome nature of the POVM case, it also allows for easy analysis of cryptographic primitives where users do not need to distinguish all $d$ basis states in a ``test'' case.  Experimentally, one need only distinguish a single basis state for the test basis and a full basis measurement, in an alternative, potentially easier to distinguish, basis, for the subsequent measurement.  That is, one need not be able to distinguish all basis states in two different bases.  This may lead to simpler cryptographic protocols and we show an example in this work.

We make several contributions in this work.  First, we derive a new quantum entropic uncertainty relation, relating conditional min entropy and the Hamming weight of a measurement outcome performed through a two-outcome measurement (regardless of the dimension of the underlying system).  Our relation is connected to the quantum sampling framework introduced in \cite{sampling} thus showing, in addition to our prior work in \cite{krawec2019quantum}, that this sampling framework has strong potential for applications in general quantum information theory while also showing a fascinating connection between classical and quantum science.  Finally, we analyze a new source-independent QRNG protocol using high dimensional quantum states, which is also potentially more practical than prior protocols in this setting.  We use our entropic uncertainty relation to prove the security of this protocol and show it can support very optimistic bit generation rates.  In fact, for many settings, our new protocol, thanks to our new entropic uncertainty relation, can actually outperform more complex protocols.  This shows the great potential benefits of using quantum sampling based entropic uncertainty relations as discussed here and in our previous work \cite{krawec2019quantum}.

\subsection{Notation}
We begin by introducing some notation and concepts we will use.  Let $\mathcal{A}_d = \{0, 1, \cdots, d-1\}$ be an alphabet of size $d$ (the exact characters do not matter so long as there is a distinguished ``$0$'' element).  Given $q \in \mathcal{A}_d^N$, and a subset $t = \{t_1, \cdots, t_m\}$ of $\{1, 2, \cdots, N\}$, we write $q_t$ to mean the substring of $q$ indexed by $t$, namely $q_t = q_{t_1}\cdots q_{t_m}$.  We use $q_{-t}$ to mean the substring of $q$ indexed by the complement of $t$.  We define the \emph{Hamming weight} of $q$ to be the number of non-zero characters in $q$.  The \emph{relative Hamming weight} of $q$, denoted $w(q)$ is the number of non-zero characters in $q$ divided by the total number of characters in $q$.  That is:
\begin{equation}
w(q) = |\{i \st q_i \ne 0\}|/|q|.
\end{equation}

A \emph{density operator} acting on Hilbert space $\mathcal{H}$ is a Hermitian positive semi-definite operator of unit trace.  Given element $\ket{\psi} \in \mathcal{H}$, we write $\kb{\psi}$ to mean the projector $\ket{\psi}\bra{\psi}$.  We use $\mathcal{H}_d$ to denote a $d$-dimensional Hilbert space.

The Shannon entropy of a random variable $X$ is denoted $H(X)$.  The $d$-ary entropy function, denoted $h_d(x)$ for $x \in [0,1]$ is defined to be:
\[h_d(x) = x\log_d(d-1) - x\log_d x - (1-x)\log_d(1-x).\]
We also define the \emph{extended $d$-ary entropy function}, denoted $\Hextd_d(x)$, for any $x \in \mathbb{R}$, as:
\begin{equation}
\Hextd_d(x) = \left\{\begin{array}{cl}
0 & \text{ if } x \le 0\\
h_d(x) & \text{ if } 0 \le x \le 1-1/d\\
1&\text{ if } x > 1-1/d
\end{array}\right.
\end{equation}

Let $\rho_{AE}$ be a density operator acting on Hilbert space $\mathcal{H}_A\otimes\mathcal{H}_E$.  Then, the \emph{conditional quantum min entropy} \cite{renner2008security}, denoted $\Hmin(A|E)_\rho$, is defined to be:
\[
\Hmin(A|E)_\rho = \sup_{\sigma_E}\max\left(\lambda\in\mathbb{R}\st 2^{-\lambda}I_A\otimes\sigma_E - \rho_{AE} \ge 0\right).
\]
Above, $I_A$ is the identity operator on $\mathcal{H}_A$ and $X\ge 0$ implies that $X$ is positive semi-definite.  If the $E$ system is trivial, it can be shown that $\Hmin(A)_\rho = -\log\lambda_{max}$, where $\lambda_{max}$ is the maximal eigenvalue of $\rho$.  If $\rho$ is a classical state (i.e., $\rho_A = \sum_xp_x\kb{x}$ for some orthonormal basis $\{\ket{x}\}$), then $\Hmin(A)_\rho = -\log\max p_x$.  The \emph{smooth min entropy}, denoted $\Hmin^\epsilon(A|E)_\rho$ is defined as \cite{renner2008security}:
\[\Hmin^\epsilon(A|E)_\rho = \sup_{\sigma\in\Gamma_\epsilon(\rho)}\Hmin(A|E)_\sigma,\]
where:
\[\Gamma_\epsilon(\rho) = \{\sigma \st \trd{\sigma-\rho} \le \epsilon\},\]
and $\trd{X}$ is the \emph{trace distance} of operator $X$.

Let $Z=\{\ket{i}\}$ be an orthonormal basis of $\mathcal{H}_A$ and let $\rho_{AE}$ be some density operator.  Then we write $\Hmin(Z|E)_\rho$ to mean the conditional min entropy of the state $\rho_{ZE}$ which results from a measurement of the $A$ system using basis $Z$.  If $\rho_{AE}$ is pure (i.e., $\rho_{AE} = \kb{\psi}$), then we may write $\Hmin(Z|E)_\psi$.  Similarly for the smooth min entropy.

Given a quantum-classical state $\rho_{AC}$ of the form $\rho_{AC} = \sum_{c=0}^Np_c\rho_A^c\otimes\kb{c}$, then it is easy to prove from the definition of min entropy that:
\begin{equation}\label{eq:cl-ent}
\Hmin(A|C)_\rho \ge \min_c\Hmin(A)_{\rho_A^c}.
\end{equation}

Min-entropy is a very useful quantity to measure and has many applications.  In quantum cryptography, one may use min-entropy to determine how many uniform independent random bits may be extracted from a quantum state.  In particular, through a \emph{privacy amplification} process, one may take as input a classical-quantum (cq) state $\rho_{AE}$ and process the $A$ register which is $N$ bits long to transform it into the cq-state $\sigma_{KE}$, where the $K$ register is $\ell$ bits long by hashing it through a two-universal hash function.  Then, as shown in \cite{renner2008security}, it holds that:
\begin{equation}\label{eq:PA}
\trd{\sigma_{KE} - I_K/2^\ell\otimes \sigma_E} \le 2^{-\frac{1}{2}(\Hmin^\epsilon(A|E)_\rho-\ell)} + 2\epsilon.
\end{equation}


An important lemma concerning min-entropy was proven in \cite{sampling} (also based on a Lemma from \cite{renner2008security}).
\begin{lemma}\label{lemma:super}
(From \cite{sampling}): Let $Z=\{\ket{i}\}$ and $X=\{\ket{x_i}\}$ be two orthonormal bases of $\mathcal{H}_Z$.  Then for any pure state $\ket{\psi} = \sum_{i\in J}\alpha_i\ket{i}\otimes\ket{\phi_i}_E \in \mathcal{H}_Z\otimes\mathcal{H}_E$ (where $\ket{\phi_i}_E$ are arbitrary, normalized, states in $\mathcal{H}_E$), if we define the mixed state $\rho = \sum_{i\in J}|\alpha_i|^2\kb{i}\otimes\kb{\phi_i}$, then:
\[\Hmin(X|E)_\psi \ge \Hmin(X|E)_\rho - \log_2|J|.\]
\end{lemma}

\section{Classical and Quantum Sampling}
As our entropic uncertainty relation is based on the quantum sampling technique introduced in \cite{sampling}, we take time here to review the relevant information.  Note that everything in this section is derived from \cite{sampling}.

Let $q \in \mathcal{A}_d^{N}$.  A \emph{sampling strategy} is a process of choosing a random subset $t \subset \{1, \cdots, N\}$ and then, given $q_t$, outputs a ``guess'' or estimate as to the value of $w(q_{-t})$.  That is, given an observation of the string $q$ indexed by $t$, the strategy will compute an estimate as to the relative Hamming weight in the unobserved portion of the string, $q_{-t}$.  In this work, we are interested in the sampling strategy that chooses $t$ of size $m$, uniformly at random and, when given $q_t$ (from a string $q \in \mathcal{A}_d^{m+n}$), will output $w(q_t)$ as a guess for $w(q_{-t})$.  We denote this strategy $\samp(d,m,n)$ (when the context is clear, we forgo writing the $m$ and $n$ parameters).

Let $B_{t,d}^\delta$ be the set of all words in $\mathcal{A}_d^{m+n}$ such that the estimate given by sampling strategy $\samp(d)$ is $\delta$ close to the actual value given a particular, fixed, subset $t$.  Formally:
\[
B_{t,d}^\delta = \{q \in \mathcal{A}_d^{n+m} \st |w(q_t) - w(q_{-t})| \le \delta\}.
\]
Then, the \emph{error probability of $\samp(d)$} is defined to be:
\[
\epsilon^{cl}_{\delta,d} = \max_{q \in \mathcal{A}^{n+m}}Pr\left(q \not\in B_{T,d}^\delta\right),
\]
where the above probability is over the choice of subset.  Note the ``$cl$'' superscript is used to enforce the notion that this is a classical sampling strategy still.  However, a classical sampling strategy may be extended to a quantum one in a natural way \cite{sampling}.  Let $Z=\{\ket{a_0}, \cdots, \ket{a_{d-1}}\}$ be an orthonormal basis of $\mathcal{H}_d$.  Then, given a state $\ket{\psi} \in \mathcal{H}_A\otimes\mathcal{H}_E$, where $\mathcal{H}_A \cong \mathcal{H}_d^{\otimes N}$, if we can write $\ket{\psi} = \ket{a_{i_1},a_{i_2},\cdots a_{i_N}}\otimes\ket{\phi}_E$, where $i=i_1\cdots i_N\in\mathcal{A}_d^{N}$, then $\ket{\psi}$ is said to have \emph{relative Hamming weight $w(i)$ in $A$ with respect to basis $Z$}.  Note that this definition is basis dependent, and not any arbitrary $\ket{\phi}_{AE}$ can be said to have Hamming weight $\beta$ using this definition - only those that are of this particular basis form.  Note we often denote $\ket{a_{i_1}\cdots a_{i_N}}$ as simply $\ket{a_i}$ if the context is clear.

Next, we define $\text{span}\left(B_{t,d}^\delta\right)$ to be $\text{span}\left(\left\{\ket{a_i} \st i \in \mathcal{A}_d^N \text{ and } |w(i_t) - w(i_{-t})|\le\delta\right\}\right)$.  Notice that if $\ket{\psi} \in \text{span}(B_{t,d}^\delta)\otimes\mathcal{H}_E$, then if sampling is done on the state $\ket{\psi}$ by measuring in the $Z$ basis on fixed subset $t$, it is guaranteed that the state will collapse to one which is a superposition of states that are $\delta$ close to the observed Hamming weight with respect to the basis used.

The main result from \cite{sampling}, besides introducing the above definitions, was to prove the following:
\begin{theorem}\label{thm:QST}
(Modified from \cite{sampling}): Let $m<n$ and consider the sampling strategy $\samp(d,m,n)$.  Then, for every pure state $\ket{\psi} \in \mathcal{H}_d^{\otimes(m+n)}\otimes\mathcal{H}_E$, there exists a collection of ``ideal states'' denoted $\{\ket{\phi}^t\}$, indexed over all subsets $t \subset \{1, \cdots, m+n\}$ of size $m$ such that $\ket{\phi^t} \in \text{span}(B_{t,d}^\delta)\otimes\mathcal{H}_E$ and:
\begin{equation}
\frac{1}{2}\trd{\frac{1}{T}\sum_t\kb{t}\otimes\kb{\psi} - \frac{1}{T}\sum_t\kb{t}\otimes\kb{\phi^t}} \le \sqrt{\epsilon^{cl}_{\delta,d}}.
\end{equation}
Above, $T = {n+m \choose m}$ and the sum is over all subsets $t$ of size $m$ and, again, $||\cdot||$ is the trace distance.
\end{theorem}
\begin{proof}
To show that the above follows from Theorem 3 in \cite{sampling}, note that, in their proof, they show that for any fixed $\ket{\psi}$, there exists a suitable ideal state satisfying the needed inequality.
\end{proof}
Actually, in \cite{sampling}, a more general statement was proven for arbitrary sampling strategies, though we focus only on $\samp(d,m,n)$ here.  We also reword their result from \cite{sampling} slightly to give a more applicable form of their result, for our work here (see also \cite{krawec2019quantum}), however the above follows immediately from the proof of their main theorem.

The following lemma, proven in \cite{sampling} will be important.
\begin{lemma}\label{lemma:samp}
(From \cite{sampling}): Let $\delta > 0$ and $d \ge 2$. Consider $\samp(d, m, n)$ for $m < n$.  Then:
$\epsilon_{\delta,d}^{cl} \le 2\exp\left(\frac{-\delta^2m(n+m)}{m+n+2}\right).$
\end{lemma}

\section{Main Result}

We are now in a position to state and prove our new entropic uncertainty relation.  We consider the following experiment, denoted $\experimentname$.  This experiment takes as input a quantum system of the form $\rho_{TAE} = \sum_tp_t\kb{t}\otimes\rho_{AE}^t$, where the sum is over all subsets $t$ of a fixed size $m$, and a two element POVM $\Lambda = \{\Lambda_0, \Lambda_1\}$.  Note that $\rho_{AE}^t$ may be equal to $\rho_{AE}^{t'}$ for $t\ne t'$ (i.e., the $AE$ portion may be independent of the $T$ register initially) and we assume the $A$ portion acts on a Hilbert space $\mathcal{H}_d^{\otimes (m+n)}$ where $d$, $m$, and $n$ are known to the experiment.  This experiment will first measure the $T$ register resulting in outcome $t$ and causing the state to collapse to $\rho_{AE}^t$.  Next, it will measure those $d$-dimensional subspaces of the $A$ register as indexed by subset $t$ using POVM $\Lambda$ resulting in outcome $q \in \{0,1\}^m$ and, then tracing out the measured portion leaving only the $n$ unmeasured subspaces of $A$ and the $E$ system, results in post-measurement state $\rho(t,q)$.  The values $t$, $q$, and the quantum state $\rho(t,q)$ are returned by the experiment.  A particular run of this experiment, with a particular output, is denoted $(t,q,\rho(t,q))\experiment{\rho_{TAE}, \Lambda}$.

Our main result involves a bound on the min entropy of the remaining system if it is measured in a $d$ dimensional basis as a function of the specific returned $q$.  With high probability, given a particular observation $q$, one may argue that the min entropy in the remaining portion, if measured in an alternative basis, may be lower bounded by a function of the basis choice and the Hamming weight of $q$.  In particular, with high probability, if the Hamming weight of $q$ is small, one may argue there is a high amount of min entropy in the remaining portion of the system if measured in an alternative basis.

\begin{theorem}\label{thm:main}
Let $\epsilon > 0$, $0 < \beta < 1/2$, and $\rho_{AE}$ an arbitrary quantum state acting on $\mathcal{H}_A\otimes\mathcal{H}_E$, where $\mathcal{H}_A \cong \mathcal{H}_d^{\otimes(n+m)}$ for $d \ge 2$ and $m < n$.  Let $Z = \{\ket{z_i}\}_{i=0}^{d-1}$ and $X = \{\ket{x_i}\}_{i=0}^{d-1}$ be two orthonormal bases of $\mathcal{H}_d$ and $\Lambda$ be the two outcome POVM with elements $\{\Lambda_0=\kb{x_0}, \Lambda_1=I-\kb{x_0}\}$ (where, $\kb{x_0} = \ket{x_0}\bra{x_0}$).  Finally, let $(t,q,\rho(t,q))\experiment{\frac{1}{T}\sum_t\kb{t}\otimes\rho_{AE}, \Lambda}$, where the sum is over all subsets $t \subset \{1,2,\cdots,n+m\}$ of size $m$ and $T = {n+m \choose m}$.  Then it holds that:
\begin{equation}
Pr\left(\Hmin^{\epsilon'}(Z|E)_{\rho(t,q)} + \frac{n\Hextd_d(w(q)+\delta)}{\log_d 2} \ge n\gamma\right) \ge 1-\epsilon'',
\end{equation}
where the probability is over the choice of subset $t$ and the measurement outcome $q$.  Above:
\[\gamma = -\log_2\max_{a,b\in\mathcal{A}_d}|\braket{z_a|x_b}|^2,
\]
and $\epsilon' = 4\epsilon+2\epsilon^\beta$,  $\epsilon'' = 2\epsilon^{1-2\beta}$ and finally:
\begin{equation}
\delta = \sqrt{\frac{(m+n+2)\ln(2/\epsilon^2)}{m(m+n)}}
\end{equation}
\end{theorem}
\begin{proof}
Our proof follows similar techniques we used first in \cite{krawec2019quantum}, though with suitable modifications for higher-dimensional systems entangled with an ancilla system.  We first consider the case where $\rho_{AE}$ is pure; that is $\rho_{AE} = \kb{\psi}$.  Consider the sampling strategy $\samp(d)$ as discussed earlier.  From Theorem \ref{thm:QST} using $\rho_{AE}$ and $\samp(d)$, we know there exits an ideal state $\sigma = \frac{1}{T}\sum_t\kb{t}\otimes\kb{\phi^t}$ such that:
\begin{enumerate}
  \item $\ket{\phi^t} \in \text{span}\left(B_{t,d}^\delta\right)\otimes \mathcal{H}_E$
  \item $\frac{1}{2}\trd{\frac{1}{T}\sum_t\kb{t}\otimes\rho_{AE} - \frac{1}{T}\sum_t\kb{t}\otimes\kb{\phi^t}} \le \sqrt{\epsilon_\delta^{cl}}.$
\end{enumerate}
From Lemma \ref{lemma:samp}, along with our choice of $\delta$, we have $\sqrt{\epsilon_\delta^{cl}} = \epsilon$.  We first analyze the ideal state $\sigma$.

Consider running $(t,q,\sigma(t,q))\experiment{\sigma, \Lambda}$.  First, the experiment will choose a random sample by measuring the $T$ register, causing $\sigma$ to collapse to the ideal $\ket{\phi^t}$.  Next, a measurement is performed using POVM $\Lambda$ resulting in outcome $q \in \{0,1\}^m$.  The experiment then traces out the measured portion resulting in $\sigma(t,q)$, a density operator acting on $\mathcal{H}_d^{\otimes n}\otimes\mathcal{H}_E$.  Since $\ket{\phi^t} \in \text{span}\left(B_{t,d}^\delta\right)\otimes\mathcal{H}_E$, we claim that the post measurement state is of the form:
\begin{equation}\label{eq:ideal-PM}
\sigma(t,q) = \sum_{k\in\al_{d-1}^{wt(q)}}p_k\cdot P\left(\sum_{i\in J_q^{(k)}}\alpha_i^{(k)}\ket{x_i}\otimes\ket{E_i^{(k)}}\right),
\end{equation}
where $P(z) = zz^*$, $\al_d$ was defined in the Notation section, $wt(q)$ is the (non-relative) Hamming weight of $q$, and:
\begin{equation}\label{eq:Jq}
J_q^{(k)} \subset J_q = \{i\in\al_d^n \st |w(i) - w(q)| \le \delta\}.
\end{equation}
That this is the form of the post measurement state after the experiment is clear.  Indeed, note that $\ket{\phi^t}$ is a superposition of vectors of the form $\ket{x_i}$ with $|w(i_t) - w(i_{-t})|\le\delta$.  Thus, on observing $q$ using POVM $\Lambda$ on subspace indexed by $t$, but before tracing out the measured portion, the state is of the form:
\[
\sum_{k\in K_q}\sqrt{p_k}\ket{x_k}_Q\sum_{i\in J_q^{(k)}}\alpha_i^{(k)}\ket{x_i}\otimes\ket{E_i^{(k)}},
\]
where $K_q = \{k\in\al_d^m \st k_i=0 \text{ iff } q_i=0\}$.  Tracing out the $Q$ register, the final step of the experiment, yields Equation \ref{eq:ideal-PM}.

We now claim that $\Hmin(Z|E)_{\sigma(t,q)} \ge n(\gamma - \Hextd_d(w(q)+\delta)/\log_d 2)$.  Consider a purification of Equation \ref{eq:ideal-PM}:
\[\ket{\sigma_{KRE}(t,q)} = \sum_k\sqrt{p_k}\ket{k}\sum_{i\in J_q^{(k)}}\alpha_i^{(k)}\ket{x_i}_R\ket{E_i^{(k)}}.\]
Then it holds that $\Hmin(Z|E)_{\sigma(t,q)} \ge \Hmin(Z|EK)_{\sigma(t,q)}$.  By reordering terms, we may write this purification as:
\[
\ket{\sigma_{KRE}(t,q)} = \sum_{i\in J_q}\beta_i\ket{x_i}\ket{\widetilde{E}_i}_{EK},
\]
where $J_q$ was defined in Equation \ref{eq:Jq} and the $\ket{\widetilde{E}_i}$ are normalized states in $\mathcal{H}_E\otimes\mathcal{H}_K$.  Define the mixed state $\chi = \sum_{i\in J_q}|\beta_i|^2\kb{i}\otimes\kb{\widetilde{E}_i}$.  Then, from Lemma \ref{lemma:super}, we have:
\[
\Hmin(Z|EK)_{\sigma(t,q)} \ge \Hmin(Z|EK)_\chi - \log_2|J_q|.
\]
We first consider a bound on $\Hmin(Z|EK)_\chi$.  After measuring in the $Z$ basis, the resulting state may be written as the density operator $\chi_{ZEK}$:
\begin{equation}
\chi_{ZEK} = \sum_{i \in J_q}|\beta_i|^2\left(\sum_{j\in\al_d^n}p(j|i)\kb{z_j}\right)\otimes\kb{\widetilde{E}_i}_{EK},
\end{equation}
where:
$p(j|i) = |\braket{z_j|x_i}|^2 = \prod_{\ell=1}^n|\braket{z_{j_\ell}|x_{i_\ell}}|^2 \le c^n,$
and $c = \max_{a,b\in\al_d}|\braket{z_a|x_b}|^2$.  We add an additional register $\mathcal{H}_I$ spanned by orthonormal basis $\{\ket{I_i}\}$ and define the state:
\begin{align*}
\chi_{ZEKI} = \sum_{i\in J_q}|\beta_i|^2\underbrace{\left(\sum_{j\in\al_d^n}p(j|i)\kb{z_j}\right)}_{\chi_{i}}\otimes\kb{\widetilde{E}_i}\otimes\kb{I_i}
\end{align*}
The $EKI$ register may be considered, taken together, as a classical system and, so, using Equation \ref{eq:cl-ent}, we have:
\begin{align*}
\Hmin(Z|EKI)_\chi &\ge \min_i\Hmin(Z)_{\chi_{i}}\\
&= \min_i\left(-\log\max_jp(j|i)\right)\\
&= -\max_{i,j}\log p(j|i)\ge -\log c^n = n\gamma.
\end{align*}
Using the well-known bound on the volume of a Hamming sphere, we have $|J_q|\le|\{i\in\al_d^n\st w(i) \le w(q)+\delta\}| \le d^{n\Hextd_d(w(q)+\delta)}$ (here, we use our extended version to avoid the case when $w(q)+\delta > 1-1/d$; indeed, in that case, the above holds trivially).  Combining everything, we conclude:
\begin{align}
\Hmin(Z|E)_{\sigma(t,q)} &\ge \Hmin(Z|EK)_{\sigma(t,q)}\notag\\
&\ge \Hmin(Z|EK)_\chi - \log_2|J_q|\notag\\
&\ge \Hmin(Z|EKI)_\chi - \log_2|J_q|\notag\\
&\ge n\left(\gamma - \frac{\Hextd_d(w(q)+\delta)}{\log_d 2}\right).
\end{align}

Of course, this was only the ideal state where the sampling process is guaranteed to produce a good result.  We now turn our attention to the real case $\rho_{AE}$.  Consider $\rho_{TQRE}$, a density operator describing the output of the experiment in its entirety, modeling the output $t$ and $q$ as random variables.  We may write this state as:
\[
\rho_{TQRE} = \frac{1}{T}\sum_t\kb{t}_T\otimes\sum_{q\in\{0,1\}^m}p(q|t)\kb{q}_Q\otimes\rho(t,q),
\]
where $p(q|t)$ is the probability of observing $q$ given that subset $t$ was chosen.  Of course $\rho(t,q)$ is the post measurement state (acting on space $RE$) output in that event, tracing out the measured portion of $\mathcal{H}_A$ (the $R$ portion is the unmeasured portion remaining after measurement).  Similarly, we may define $\sigma_{TQRE}$ to be the result of the entire experiment performed on the ideal state:
\[
\sigma_{TQRE} = \frac{1}{T}\sum_t\kb{t}_T\otimes\sum_{q\in\{0,1\}^m}\hat{p}(q|t)\kb{q}_Q\otimes\sigma(t,q).
\]
Of course, $\sigma(t,q)$, the post measurement state for the ideal scenario, was analyzed above.

Since quantum operations, in particular our experiment, cannot increase trace distance, we have $\frac{1}{2}\trd{\rho_{TQRE} - \sigma_{TQRE}} \le \epsilon$.  Let $\delta_{t,q} = \hat{p}(q|t) - p(q|t)$.  By elementary properties of trace distance, we have:
\begin{align*}
\epsilon &\ge \frac{1}{2}\trd{\rho_{TQRE} - \sigma_{TQRE}}\\
&= \frac{1}{2}\sum_t\frac{1}{T}\sum_q\trd{p(q|t)\rho(t,q) - \hat{p}(t,q)\sigma(t,q)}\\
&=\frac{1}{2T}\sum_{t,q}\trd{p(q|t)(\rho(t,q) - \sigma(t,q)) - \delta_{t,q}\sigma(t,q)}\\
&\ge\sum_{t,q}p(q\wedge t)\frac{1}{2}\trd{\rho(t,q)-\sigma(t,q)} - \sum_{t,q}\frac{1}{2T}\trd{\delta_{t,q}\sigma(t,q)}\\
&=\sum_{t,q}p(q\wedge t)\Delta_{t,q} - \sum_{t,q}\frac{1}{2T}|\delta_{t,q}|,
\end{align*}
where we define $p(q\wedge t) = \frac{1}{T}p(q|t)$ and $\Delta_{t,q} = \frac{1}{2}\trd{\rho(t,q)-\sigma(t,q)}$.  The above follows from the reverse triangle inequality and the fact that $\trd{\sigma(t,q)} = 1$ since $\sigma(t,q)$ is a positive operator of unit trace. Note that $\Delta_{t,q} \le 1$ due to properties of trace distance.

Since partial trace is a quantum operation, we have (tracing out the $RE$ registers):
$\epsilon \ge \frac{1}{2}\trd{\rho_{TQ}-\sigma_{TQ}} = \sum_{t,q}\frac{1}{2T}|\delta_{t,q}|.$
Combining the above yields:
$\sum_{t,q}p(q\wedge t) \Delta_{t,q} \le 2\epsilon.$
Now, we treat $\Delta_{t,q}$ as a random variable over the choice of subset ($t$) and measurement outcome ($q$).  It is clear that the expected value of $\Delta_{t,q}$ is $\mathbb{E}(\Delta_{t,q}) = \mu \le 2\epsilon$.  The variance, $V^2$, is also bounded by:
\begin{align*}
V^2 &= \sum_{t,q}p(q\wedge t)\Delta_{t,q}^2 - \mu^2 \le \sum_{t,q}p(q\wedge t)\Delta_{t,q}\le2\epsilon\notag\\
\end{align*}
The above follows from the fact that $\Delta_{t,q} \le 1$.  By Chebyshev's inequality, we have:
$Pr\left(|\Delta_{t,q} - \mu| \le \epsilon^\beta\right) \ge 1-2\epsilon^{1-2\beta}.$
Thus, except with probability at most $2\epsilon^{1-2\beta}$, it holds that:
$|\Delta_{t,q}-\mu| \le \epsilon^\beta\Longrightarrow \frac{1}{2}\trd{\rho(t,q) - \sigma(t,q)} \le 2\epsilon+\epsilon^\beta.$
Since, in such a case, $\sigma(t,q) \in \Gamma_{4\epsilon+2\epsilon^\beta}(\rho(t,q))$, we conclude:
\begin{align*}
\Hmin^{4\epsilon+2\epsilon^\beta}(Z|E)_{\rho(t,q)} &\ge \Hmin(Z|E)_{\sigma(t,q)}\\
&\ge n\left(\gamma - \frac{\Hextd(w(q)+\delta)}{\log_d 2}\right),
\end{align*}
as desired.

Of course, if $\rho_{AE}$ is not pure, it may be purified by adding an ancilla system $\mathcal{H}_{I}$.  In that case, due to strong sub additivity, the above analysis still holds, thus completing the proof.

\end{proof}

\section{Application to QRNGs}

While interesting in itself, our new entropic uncertainty relation has applications to cryptography.  Note that we consider the main contribution of this paper to be our Theorem \ref{thm:main}, however, in this section, we show how it can be used in applications.

In particular, we use it now to demonstrate the security of the following \emph{source independent} quantum random number generator (QRNG).  The goal of a QRNG is to utilize quantum effects to distill a truly  uniform random string.  The source independent model, introduced in \cite{si-qrng-first} assumes the quantum source is controlled by an adversary (though the dimension of the system is known and bounded) while the measurement devices are trusted.  Furthermore, in this model, the goal is to produce a uniform random string, independent of any adversary's system.  The protocol we analyze is the following: a source, potentially adversarial, produces a quantum state in $\mathcal{H}_d^{\otimes (n+m)}\otimes\mathcal{H}_E$ where $d$, $m$, and $n$ are public parameters set by the users of the protocol.  The $n+m$ qudits are sent to the user Alice, while the $\mathcal{H}_E$ system is kept by the adversary.  Alice chooses a subset of size $m$ qudits to measure using POVM $\Lambda=\{\kb{x_0}, I-\kb{x_0}\}$ where $\ket{x_0} = \mathcal{F}\ket{0}$, and $\mathcal{F}$ is the $d$ dimensional quantum Fourier transform.  The remaining $n$ qudits are measured in the computational $Z=\{\ket{0},\cdots,\ket{d-1}\}$ basis resulting in a string $r$.  This is then processed through privacy amplification to hash $r$ down to an $\ell$ bit string $s$ which is the final random string output by the protocol.  Note that, an honest source should prepare a state of the form $\ket{x_0}^{\otimes(m+n)}$, independent of $\mathcal{H}_E$. To our knowledge this source independent QRNG has not been considered in the past.  Indeed, prior work in this model requires the user to be able to perform a full basis measurement both for the test and the random distillation modes.  \emph{Thus, our protocol would be simpler to implement in practice (as one need not distinguish all states in two bases).}

Let $\epsilon > 0$ and set $\epsilon_{PA} = 9\epsilon+4\epsilon^\beta$ be the desired distance from an ideal uniform random string of size $\ell$ independent of $E$'s system.  Using Equation \ref{eq:PA} and Theorem \ref{thm:main}, after running the protocol, on observing outcome $q$ during the test with $\Lambda$, except with probability $2\epsilon^{1-2\beta}$, it holds that:
%
\begin{equation}
\ell_{ours} \ge n\left(\log d - \frac{\Hextd(w(q)+\delta)}{\log_d 2}\right) - 2 \log\frac{1}{\epsilon},
\end{equation}
giving a simple, clean, proof of security for this new protocol.  Thus, to analyze the number of random bits one may distill from the protocol we introduced above, one simply observes $q$ using a test of POVM $\Lambda$ \emph{which does not require a full basis measurement}.  From this, one may, with high probability depending on user parameters, determine how many random bits are output even if the source is adversarial.

We compare with two other high dimensional source independent QRNG's - one from \cite{si-qrng-first} (with bit generation length $\ell_1$ as derived in \cite{si-qrng-first}) and one from \cite{xu2016experimental} (with bit generation length $\ell_2$ as derived in \cite{xu2016experimental}).  Both use alternative entropic uncertainty relations to compute $\ell_i$.  \emph{Note that both also require full basis measurements for testing.}


For the protocol in \cite{si-qrng-first}, an adversarial source prepares a state in $\mathcal{H}_d^{\otimes(n+m)}\otimes\mathcal{H}_E$.  Alice measures a subset in the $X=\{\ket{x_i}\}$ basis where $\ket{x_i} = \mathcal{F}\ket{i}$.  The remaining qudits are measured in the computational $Z$ basis and are processed through privacy amplification.  The secret random string size is computed in \cite{si-qrng-first} to be:
\[
\ell_{1} \ge n\left(\log_2 d - 2\log_2\left[ \frac{\Gamma(m+d)}{\Gamma\left(m+d+\frac{1}{2}\right)}\sum_{i=0}^{d-1}\frac{\Gamma\left(c_i + \frac{3}{2}\right)}{\Gamma\left(c_i + 1\right)}\right]\right),
\]
where $c_i$ is the number of measurement outcomes (out of the $m$ test measurements) resulting in outcome $\ket{x_i}$ and $\Gamma(x)$ is the Gamma function.  To derive the above, they used an entropic uncertainty relation from \cite{smooth-uncertainty}, along with the Bayesian estimator for the max entropy from \cite{holste1998bayes}.

The protocol introduced in \cite{xu2016experimental} involves an adversarial source that prepares an entangled pair of qudits, sending both pairs to Alice.  On test iterations, Alice measures both pairs in the basis $X$ (as defined above).  On other iterations, she measures only the first pair in basis $Z$, discarding the second pair.  Again, the authors use an entropic uncertainty relation from \cite{smooth-uncertainty}, though an alternative method of estimating the max entropy using results in \cite{max-ent-bound} and the fact that the source is preparing entangled pairs.  They prove the secret random string length, after privacy amplification, is:
\[
\ell_2 \ge n\log_2d - \log_2\gamma(d_0 + \delta'),
\]
where:
\[
\gamma(x) = (x+\sqrt{1+x^2})\left(\frac{x}{\sqrt{1+x^2}-1}\right)^x,
\]
and:
\[
\delta' = d\sqrt{\frac{N^2}{n^2m}\ln\left(\frac{4}{\epsilon'}\right)}.
\]
Above, $d_0 = \frac{1}{m}\sum_{i=1}^m|c_A(i) - c_B(i)|$, where $c_A(i)\in\mathcal{A}_d$ is measurement outcome on test iteration $i$ of the $A$ register in basis $X$ (similar for $c_B(i)$).

To evaluate our protocol ($\ell_{ours}$), we set $\beta = 1/3$ and $\epsilon = 10^{-36}$ which implies the failure probability is $2\times 10^{-12}$ while $\epsilon_{PA} = 4\times10^{-12}$. Note we did not optimize $\beta$ which may lead to higher rates for our protocol and we use $7\%$ of total signals for sampling.    When considering noise of $x$ in these evaluations we assume a depolarization channel.  For this, we set $q = x$ for our model; for $\ell_1$ we set $c_i = m\cdot x/(d-1)$ if $i \ne 0$ and $c_0 = m(1-x)$; and finally for $\ell_2$, we set $d_0 = x$ (which is advantageous for that model; indeed $x$ is only a lower-bound for $d_0$ so $\ell_2$ may be lower than we plot here). A more detailed comparison for other noise channels would be interesting future work.

The results are shown in Figure \ref{fig:1}.  We find that, for very few signals, $\ell_1$ outperforms both while for a very large number of signals, $\ell_2$ outperforms both.  However there is a large window in between where our new protocol, as analyzed by our new entropic uncertainty relation, outperforms both systems, \emph{even though we actually have a simpler protocol.}

\begin{figure}
  \centering
  \includegraphics[width=0.45\linewidth]{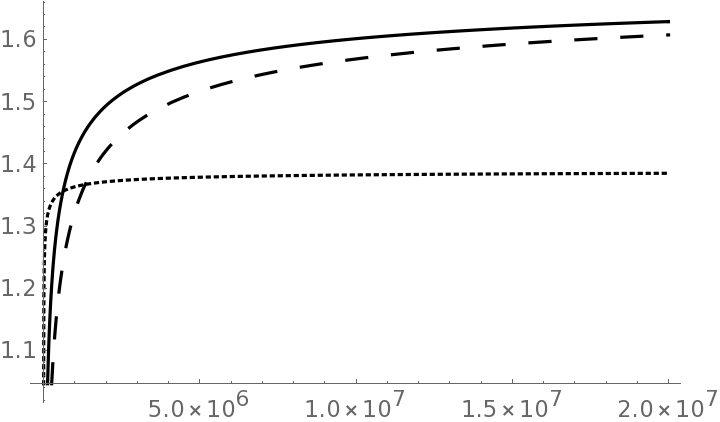}
  \includegraphics[width=0.45\linewidth]{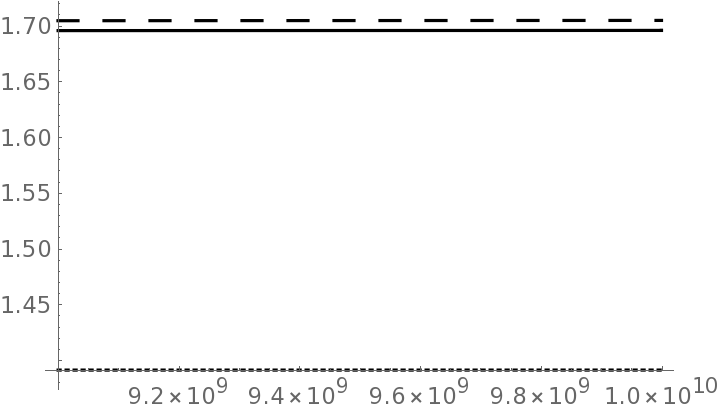}

  \includegraphics[width=0.45\linewidth]{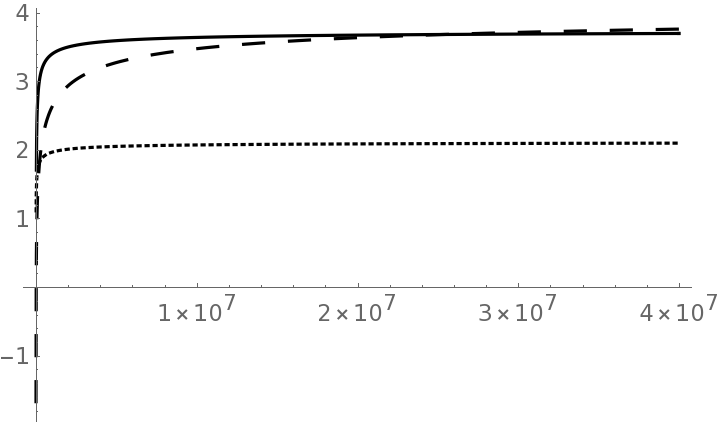}
  \includegraphics[width=0.45\linewidth]{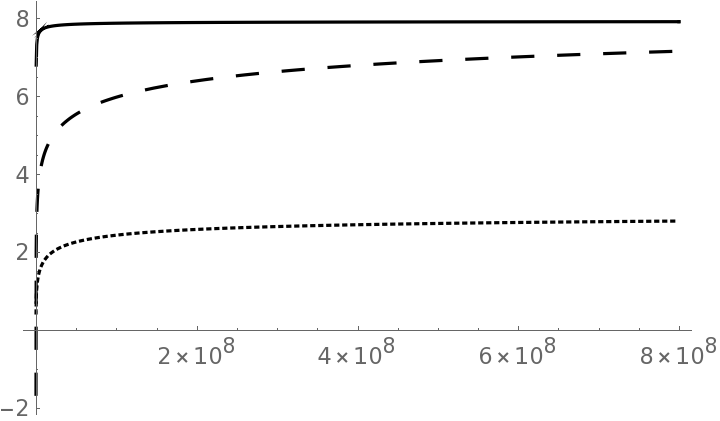}

\caption{Secret random bit generation rates. $x$-axis: Total number of signals $N=n+m$; $y$-axis: Secret random bit generation rate: $\ell/N$. Solid: ours ($\ell_{ours}/N$); Dotted: $\ell_1/N$ from \cite{si-qrng-first}; Dashed: $\ell_2/N$ from \cite{xu2016experimental}. Upper-left: $d = 2^2$ with $2\%$ noise; Upper-Right: $d=2^2$ with $2\%$ noise, higher number of iterations; Lower-Left: $d=2^5$ with $10\%$ noise; Lower-Right: $d=2^{10}$ with $10\%$ noise.  See text for explanation.}\label{fig:1}
\end{figure}

\section{Closing Remarks}

In this paper, we introduced a novel entropic uncertainty relation bounding the conditional min-entropy of a system based on the result of a measurement in a two-outcome POVM and the probability of failure of a \emph{classical} sampling strategy.  Furthermore, this shows yet another fascinating application of the quantum sampling framework as introduced in \cite{sampling} to areas in general quantum information theory.  While interesting in and of itself, we also showed how this could be used to analyze the security of a novel source independent QRNG utilizing restricted measurement capabilities.  We show our new uncertainty relation provides optimistic bit generation rates for our protocol, despite its inability to perform a complete measurement in two bases.  We believe the quantum sampling framework can hold even further applications when combined with our proof technique here and in \cite{krawec2019quantum}, and may shed light on new min entropy bounds of great use in quantum cryptography.

\section*{Acknowledgment}
The author would like to acknowledge support from NSF grant number 1812070.

\balance


\end{document}